\documentclass[a4paper,UKenglish]{article}

\usepackage{graphicx,amsmath,xcolor,amssymb}
\usepackage{xspace}
 \usepackage{authblk}
\usepackage{fullpage,verbatim}
\newcommand{\jyoti}[1]{\textcolor{black}{#1}}
\usepackage{amsthm}
\usepackage{amsmath}
    \newtheorem{theorem}{Theorem}
    \newtheorem{corollary}[theorem]{Corollary}
    \newtheorem{lemma}[theorem]{Lemma}

\title{Minimum Ply Covering of Points with Unit Squares\footnote{This work is supported in part by the Natural Sciences and Engineering Research Council of Canada (NSERC).}} %TODO Please add

\author[1]{Stephane Durocher}
\author[2]{J. Mark Keil}
\author[2]{Debajyoti Mondal}
\affil[1]{University of Manitoba, Winnipeg, Canada\\

  \texttt{durocher@cs.umanitoba.ca}} 
\affil[2]{University of Saskatchewan, Saskatoon, Canada\\

  \texttt{keil@cs.usask.ca, dmondal@cs.usask.ca}}

\newcommand{\ply}{\operatorname{ply}}
\newcommand{\plycover}{\operatorname{plycover}}

\newcommand{\mpcn}{minimum ply cover number\xspace}

\DeclareMathOperator*{\argminA}{arg\,min} % Jan Hlavacek

\begin{document}

\maketitle

\begin{abstract}
Given a set $P$ of points and a set $U$ of axis-parallel unit squares in the Euclidean plane, a minimum ply cover of $P$ with $U$ is a subset of $U$ that covers $P$ and minimizes the number of squares that share a common intersection, called the minimum ply cover number of $P$ with $U$. Biedl et al.~[Comput. Geom., 94:101712, 2020] showed that determining the minimum ply cover number for a set of points by a set of axis-parallel unit squares is NP-hard, and gave a polynomial-time 2-approximation algorithm for instances in which the minimum ply cover number is constant. The question of whether there exists a polynomial-time approximation algorithm remained open when the minimum ply cover number is $\omega(1)$. We settle this open question and present a polynomial-time \jyoti{$(8+\varepsilon)$-approximation algorithm for the general problem, for every fixed $\varepsilon>0$}.
\end{abstract}

\section{Introduction}
The \emph{ply} of a set $S$, denoted $\ply(S)$, is the maximum cardinality of any subset of $S$ that has a non-empty common intersection. The set $S$ \emph{covers} the set $P$ if $P \subseteq \bigcup_{S_i \in S} S_i$. Given sets $P$ and $U$, a subset $S \subseteq U$ is a \emph{minimum ply cover} of $P$ if $S$ covers $P$ and $S$ minimizes $\ply(S)$ over all subsets of $U$. Formally:
\begin{equation}\label{eqn:plyCoverDef}
    \plycover(P,U) = \argminA_{\substack{S\subseteq U\\ \text{$S$ covers $P$}}}  \ply(S).
\end{equation}
%The ply number of a \mpc of $P$ called the \emph{\mpcn} of $P$.  

The ply of such a set $S$ is called the \emph{minimum ply cover number} of $P$ with $U$, denoted $\ply^*(P,U)$. Motivated by applications in covering problems, including interference minimization in wireless networks, Biedl et al.~\cite{DBLP:journals/comgeo/BiedlBL21} introduced the \emph{minimum ply cover problem}: given sets $P$ and $U$, find a subset $S \subseteq U$ that minimizes \eqref{eqn:plyCoverDef}. They showed that the problem is NP-hard to solve exactly, and remains NP-hard to approximate by a ratio less than two when $P$ is a set of points in $\mathbb{R}^2$ and $U$ is a set of axis-aligned unit squares or a set of unit disks in $\mathbb{R}^2$. They also provided 2-approximation algorithms parameterized in terms of $\ply^*(P,U)$ for unit disks and unit squares in $\mathbb{R}^2$. Their algorithm for axis-parallel unit squares runs in $O((k+|P|) (2\cdot |U|)^{3k+1})$ time, where $k=\ply^*(P,U)$, which is polynomial when $\ply^*(P,U) \in O(1)$. Biniaz and Lin~\cite{DBLP:conf/cccg/BiniazL20} generalized this result for any fixed-size convex shape and obtained a 2-approximation algorithm  when $\ply^*(P,U) \in O(1)$. The problem of finding a polynomial-time approximation algorithm to the minimum ply cover problem remained open when the minimum ply cover number, $\ply^*(P,U)$, is not bounded by any constant. \jyoti{This open problem is relevant to the motivating application of interference minimization. For example, algorithms for constructing a connected network on a given set of wireless nodes sometimes produce a network with high interference~\cite{DBLP:conf/cccg/DurocherM17}. Selecting a set of network hubs that minimizes interference relates to the ply covering problem in a setting where ply may not be a constant.} 
 
Given a set $P$ and a set $U$ of subsets of $P$, the \emph{minimum membership set cover problem}, introduced by Kuhn et al.~\cite{kuhn2005}, seeks to find a subset $S \subseteq U$ that covers $P$ while minimizing the maximum number of elements of $S$ that contain a common point of $P$. A rich body of research examines the minimum membership set cover problem (e.g., \cite{DBLP:journals/siamcomp/DemaineFHS08,DBLP:journals/algorithmica/MisraMRSS13}). The minimum ply cover problem is a generalization of the minimum membership set cover problem: $U$ is not restricted to subsets of $P$, and ply is measured at any point covered by $U$ instead of being restricted to points in $P$. Consequently, the cardinality of a minimum membership set cover is at most the cardinality of a minimum ply cover. Erlebach and van Leeuwen~\cite{DBLP:conf/soda/ErlebachL08} showed that the minimum membership set cover problem remains NP-hard when $P$ is a set of points in $\mathbb{R}^2$ and $U$ are unit squares or unit disks. For unit squares, they gave a 5-approximation algorithm for instances where the optimum objective value is bounded by a constant.  Improved approximation algorithms are found in~\cite{DBLP:journals/dmaa/BasappaD18} and~\cite{DBLP:conf/approx/ErlebachL10}. A rich body of literature studies geometric set cover problems. We refer the readers to~\cite{DBLP:series/lncs/AgarwalEF19,DBLP:conf/compgeom/ChanH20} for more details on  geometric set cover problems.

\begin{figure}[pt]
    \centering
    \includegraphics[width=\textwidth]{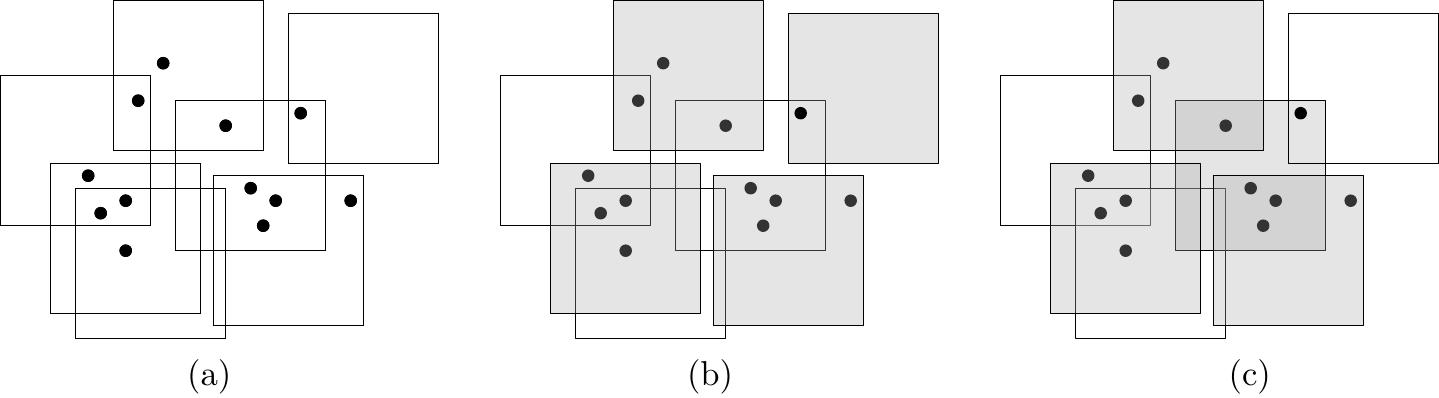}
    \caption{(a) An input consisting of points and unit squares. (b) A covering of the points with ply 1, which is also the \mpcn for the given input. (c)  A covering of the points with ply 2.}
    \label{fig:intro}
\end{figure}

\textbf{Our contribution:} In this paper we consider the minimum ply cover problem for a set $P$ of points in $\mathbb{R}^2$ with a set $U$ of axis-aligned unit squares in $\mathbb{R}^2$. We show that \jyoti{for every fixed $\varepsilon>0,$} the \mpcn can be approximated in polynomial time for unit squares within a factor of \jyoti{$(8+\varepsilon)$}. This settles an open question posed in~\cite{DBLP:journals/comgeo/BiedlBL21}  and~\cite{DBLP:conf/cccg/BiniazL20}. 

\jyoti{Our algorithm overlays a regular grid on the plane and then approximates the ply cover number  from the exact solutions  for these grid cells. The most interesting part of the algorithm is to model the idea of bounding the ply cover number with a set of budget points, and to exploit this set's geometric properties to enable dynamic programming to be applied. We show that one can set budget at the corners of a grid cell and check for a solution where the number of squares hit by a corner does not exceed its assigned budget.   A major challenge to solve this decision problem is that the squares that hit the four corners may mutually intersect to create a ply that is bigger than any budget set at the corners. This  makes it difficult to combine the solutions to the subproblems while designing a dynamic program. We show that an optimal solution can take a few well-behaved forms that can be leveraged to describe the subproblems concisely and to  combine   subproblem solutions to answer the decision problem.} 

%and unit circles, which settles the open question posed by ??\cite{??}.   For unit squares, we give an 8-approximation $O(?)$-time algorithm and for unit disks, we give a 64-approximation $O(?)$-time algorithm.

%\section{Preliminaries}

\section{Minimum Ply Covering with Unit Squares}
\label{unitsquare}
Let $P$ be a set of points in $\mathbb{R}^2$ and let $U$ be a set of axis-parallel squares in general position in $\mathbb{R}^2$. I.e., no two squares in $U$ have edges that lie on a common vertical or horizontal line. In this section we give a polynomial-time algorithm to approximate the \mpcn for $P$ with $U$. The idea of the algorithm is as follows. 

We consider a unit grid $\mathcal{G}$ over the point set $P$. The rows and columns of the grid are aligned with the $x$- and $y$-axes, respectively, and each cell of the grid is a unit square. We choose a grid that is in general position relative to the squares in $U$. A grid cell is called \emph{non-empty} if it contains some points of $P$. We prove that one can first find a small ply cover for each non-empty grid cell $R$ and then combine the solutions to obtain an approximate solution for $P$. We only focus on the ply inside $R$, because if the ply of a minimum ply cover is realized outside $R$, then there also exists a point inside $R$ giving the same ply number.

We first show how to find a small ply cover when the points are bounded inside a unit square (Section~\ref{sq}), and then show how an approximate ply cover number can be computed for $P$ (Section~\ref{gen}).

\subsection{Points are in a Grid Cell}
\label{sq}
Let $R$ be a $1 \times 1$ grid cell. Let $Q \subseteq P$ be the set of points in $R$, and let $W\subseteq U$ be the set of squares that intersect $R$. Note that by the construction of the grid $\mathcal{G}$, every square in $W$ contains exactly one corner of $R$. %We order the squares with respect to the $y$-coordinates of their bottom boundaries. In other words, a square $A$ is higher than another square $B$ if the bottom boundary of $A$ has a higher $y$-coordinate than that of $B$. 
We   distinguish some cases depending on the position of the squares in $W$. In each case we show how to compute either the minimum ply cover or a ply cover of size at most four more than the minimum ply cover number in polynomial time.

%{\color{red} We need to tell that there must be a point inside  $R$ that determines the ply.}
%Let $S\subseteq U$ be the minimum ply cover of $Q$. %By a \emph{ply point}  of $S$ we denote a point that acts as a witness of the ply number, i.e., it intersects exactly $ply(S)$ squares in $S$. We now  consider a few cases depending on the potential locations of a ply point. 

\subsubsection*{Case 1 (A corner of $R$ intersects all squares in $W$)} In this case we compute a  minimum ply cover. Without loss of generality assume that the top right corner of $R$ intersects all the squares in $W$.  We now can construct a minimum ply cover by %repeating the following {\color{red}greedy} steps.
 the following greedy algorithm.

\begin{enumerate}
    \item[] Step 1:  Let $z$ be the   leftmost uncovered point of $Q$. Find the square $B\in W$  with the lowest bottom boundary among the squares that contain $z$ (break ties arbitrarily). 
    \item[] Step 2: Add $B$ to the solution,  remove the points covered by $B$. % and remove the squares that do not contain any uncovered point of $Q$.
    
    \item[] Step 3: Repeat Steps 1 and 2 unless all the points are covered.
\end{enumerate}

\begin{figure}
    \centering
    \includegraphics[width=\textwidth]{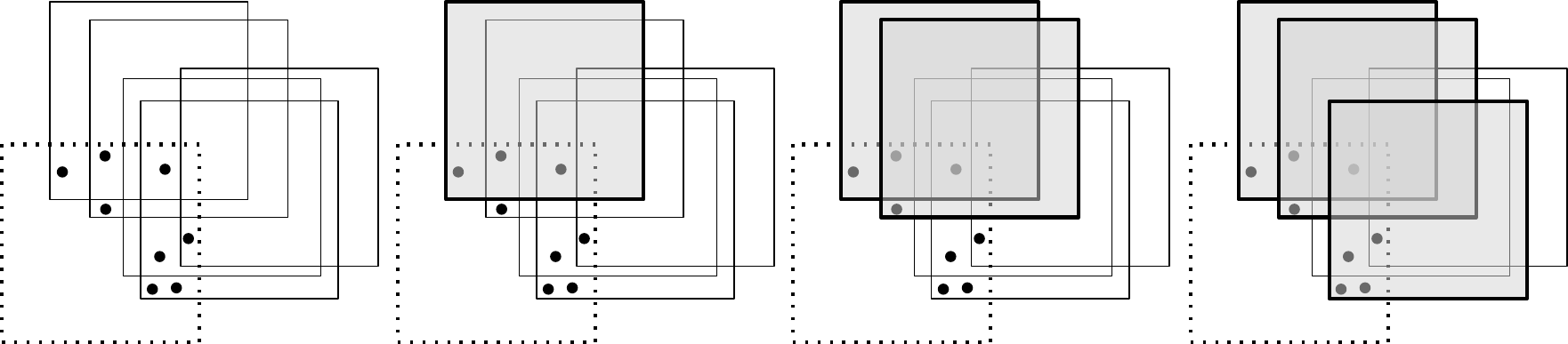}
    \caption{Illustration for Case 1. The squares taken in the solution are shaded in gray.}
    \label{fig:case1}
\end{figure}
Figure~\ref{fig:case1} illustrates such an example for Case 1. It is straightforward to compute such a solution in $O((|W|+|Q|)\log^2(|W|+|Q|))$ time using standard dynamic data structures, i.e, the point $z$ can be maintained using a range tree and the square $B$ can be maintained by leveraging dynamic segment trees~\cite{DBLP:journals/jacm/KreveldO93}.
%{\color{red}Need a proof?}

To verify the correctness of the greedy algorithm, first observe that in this case the number of squares in a minimum cardinality cover coincides with the minimum ply cover. We now show that the above greedy algorithm constructs a minimum cardinality cover. We employ an induction on the number of squares in a minimum cardinality cover.  Let  $W_1,W_2,\ldots,W_k$ be a  set of squares in the minimum cardinality cover. First consider the base case where $k=1$. Since $W_1$ covers all the points, it also covers $z$. Since $z$ is the leftmost point and since our choice of square $B$ has the lowest bottom boundary, $B$ must cover all the points. Assume now that if the minimum cardinality cover contains less than $k$ squares, then the greedy algorithm constructs a minimum cardinality cover. Consider now the case when we have $k$ squares in the minimum cardinality cover. Without loss of generality let $W_1$ be a square that contains $z$. Then any point covered by $W_1$ would also be covered by the greedy choice $B$. We can now apply the induction hypothesis to construct a cover of size at most $(k-1)$ for the points not covered by $B$.

\subsubsection*{Case 2 (Two consecutive corners of $R$ intersect  all the squares in $W$)} 
In this case we compute a  minimum ply cover. Without loss of generality assume that the top left and top right corners of $R$ intersect all the squares in $W$. Let $W_\ell$ and $W_r$ be the squares of $W$ that intersect the top left corner and top right corner, respectively. We   construct a minimum ply cover by considering whether a square of $W_\ell$ intersects a square of $W_r$.%repeating the following steps.

If the squares of $W_\ell$ do not intersect the squares of $W_r$, then we can reduce it into two subproblems of type Case 1. We solve them  independently and it is straightforward to observe that the resulting solution yields a minimum ply cover. Similar to Case 1, here we need $O((|W|+|Q|)\log^2(|W|+|Q|))$ time. 
Consider now the case when some squares in  $W_\ell$  intersect  some squares of $W_r$. To solve this case, we consider the following property of a minimum ply cover. %\jyoti{We include  the proof in Appendix~\ref{app:lemma1} due to space constraint.}

\begin{lemma}\label{lem:disjoint}
Let $S$ be a minimum ply cover of the  points in $R$ such that every square in $S$ is necessary. In other words, if a square of $S$ is removed, then the resulting set cannot cover all the points of $R$. Then %there exists a minimum ply cover $S\subseteq W$ of $R$ such that 
one can remove at most one square from $S$ to ensure that the squares of $S\cap W_\ell$ do not intersect the squares in $S\cap W_r$, where $W_\ell$ and $W_r$ are the squares that contain the top left and top right corners of $R$, respectively. 
\end{lemma}
\begin{proof}
If there exists only one square in $S\cap W_\ell$ that intersects some squares in $S\cap W_r$ (or, vice versa), then we can remove that square and the claim holds.

Suppose now for a contradiction that there exist at least two squares in $S\cap W_\ell$ that intersect some squares in $S\cap W_r$. Let  $A$ and $B$ be the topmost such squares in $S\cap W_\ell$, 
i.e., each of $A$ and $B$ intersects one or more  squares from $S\cap W_r$. Since $A$ and $B$ belong to $S$ and since every square in $S$ is necessary for the cover, each of $A$ and $B$ must cover at least one point that is not covered by any other square in $S$. Therefore, without loss of generality we can assume that the right boundary of $A$ intersects the bottom boundary of $B$,  e.g., Figure~\ref{fig:2-corner}(a). Let $r$ be the point of intersection. 

If only one square $C$ in $S\cap W_r$ intersects $B$ and includes $r$, then we can safely remove $B$ from $S$,  e.g., Figure~\ref{fig:2-corner}(a). This contradicts our assumption that all the squares of $S$ are necessary. 

If two squares $C$ and $D$ in $S\cap W_r$ intersect $B$ and $B$ contains their intersection point, then at least one of $C$ and $D$ can be safely removed,  e.g., Figure~\ref{fig:2-corner}(b).  This contradicts our assumption that all the squares of $S$ are necessary.

If two squares $C$ and $D$ in $S\cap W_r$ intersect $B$ and $B$ does not contain their intersection point, then we consider the following  scenario. If $A$ intersects $C$ or $D$, then that would make $B$ unnecessary. Therefore, we may assume that $A$ intersects some other square $E$ from $S\cap W_r$, as shown in dashed in Figure~\ref{fig:2-corner}(c)--(d). But then $B$ and $E$ cannot both be necessary. Therefore, such a scenario cannot appear. %does not intersect  removing $B$ from $S$ will ensure that the squares of $(S\cap W_\ell)$ do not intersect the squares in   $(S\cap W_r)$,  e.g., Figure~\ref{fig:2-corner}(c). 

The remaining case is when only one square $C$ in $S\cap W_r$ intersects $B$ where $C$ does not  contain $r$. Note that  $A$ has to intersect a square from $S\cap W_r$ and $B$ cannot intersect two squares from $S\cap W_r$. Therefore,  $A$ must  intersect $C$, e.g., Figure~\ref{fig:2-corner}(e). Here, removing $B$ from $S$ will ensure that the squares of $S\cap W_\ell$ do not intersect the squares in $S\cap W_r$. Since $B$ is necessary, this will leave some points uncovered. However, the lemma statement only asks for a partition of the squares and that property holds.
%Hence there can be at most one square $C$ in $(S\cap W_r)$ that intersects $B$ but does not include $r$. Therefore, removing $C$ from $S$ will ensure that the squares of $(S\cap W_\ell)$ do not intersect the squares in   $(S\cap W_r)$. 
\end{proof}
\begin{figure}[h]
    \centering
    \includegraphics[width=\textwidth]{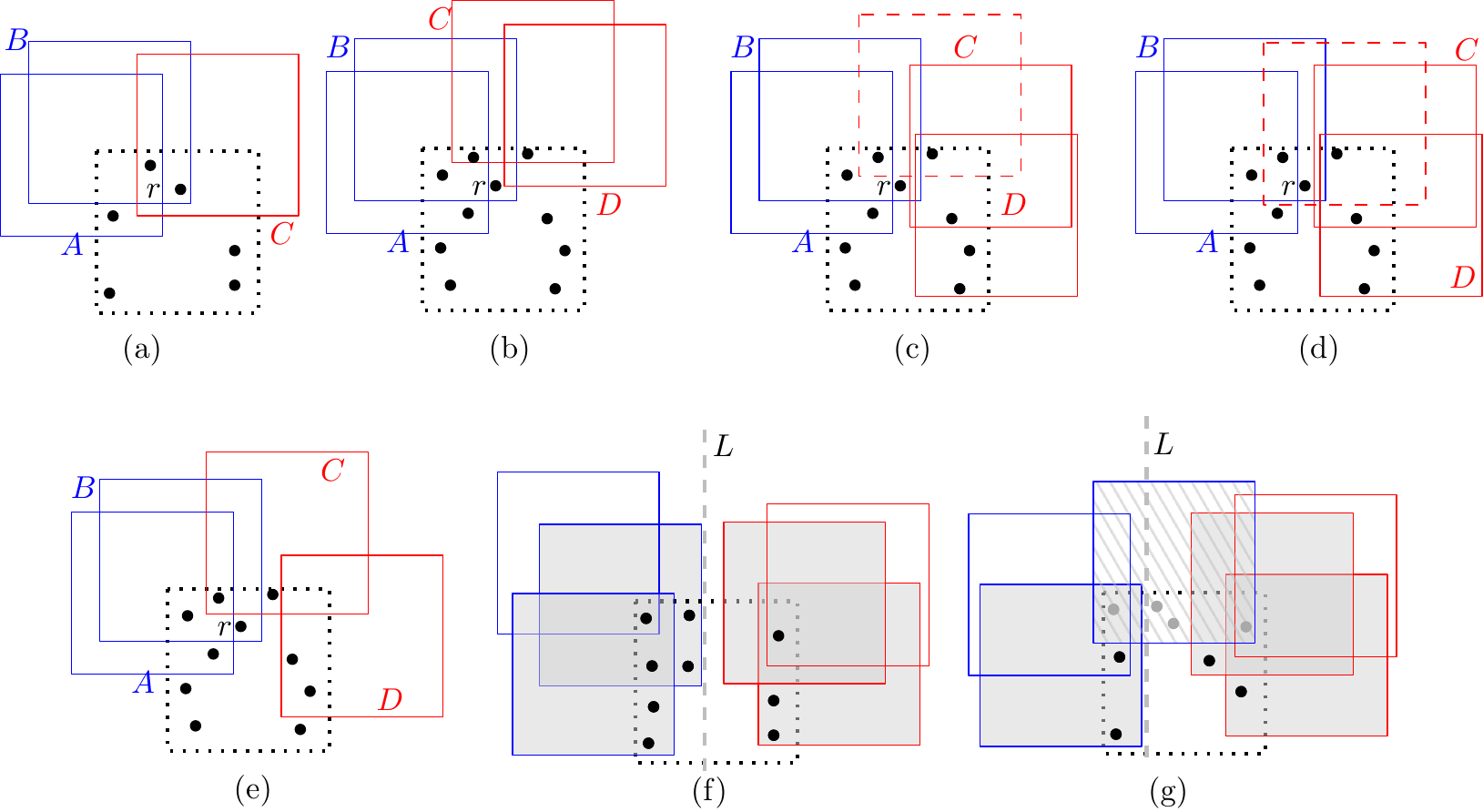}
    \caption{(a)--(e) Illustration for Lemma~\ref{lem:disjoint}, where 
    $(S\cap W_\ell)$ and   $(S\cap W_r)$ are shown in blue and red, respectively. (f)--(g) Illustration for the properties $C_1$ and $C_2$.}
    \label{fig:2-corner}
\end{figure}

%\begin{figure}[pt]
%    \centering
%    \includegraphics[width=.6\textwidth]{figures/extra.pdf}
%    \caption{(a)--(b)   Illustration for the properties $C_1$ and $C_2$.}
%    \label{fig:extra}
%\end{figure}
By Lemma~\ref{lem:disjoint}, there exists a minimum ply cover $S$ such that at least one of the following two properties hold:

\begin{enumerate}
    \item[$C_1$] There exists a vertical line $L$ that passes through the left or right side of some square  and  separates $S\cap W_\ell$ and $S\cap W_r$, as illustrated in Figure~\ref{fig:2-corner}(f). 
    \item[$C_2$] There exists a square $M$ in $S$ such that after the removal of $M$ from $S$, one can find a vertical line $L$  that passes through the left or right side of some square and   separates $S\cap W_\ell$ and $S\cap W_r$. This is illustrated in Figure~\ref{fig:2-corner}(g), where the square $M$ is shown with the falling pattern. % along with the points it covers are removed. 
\end{enumerate}

To find the minimum ply cover, we thus try out all possible $L$ (for $C_1$), and all possible $M$ and $L$ (for $C_2$). More specifically, to consider $C_1$,  for each vertical line $L$ passing through the left or right side of some square in $W$, we independently find the minimum ply cover for the points and squares on the left halfplane of $L$ and right halfplane of $L$. We then construct a ply cover of $Q$ by taking the union of these two minimum ply covers.  

To consider $C_2$, for each square $M$, we first delete $M$ and the points it covers. Then for each vertical line $L$ determined by the squares in $(W\setminus M)$, we independently find the minimum ply cover for the points and squares on the left halfplane of $L$ and right halfplane of $L$. We then construct a ply cover of $Q$ by taking the union of these two minimum ply covers and $M$.  Finally, among all the ply covers constructed considering $C_1$ and $C_2$, we choose the ply cover with the minimum ply as the minimum ply cover of $Q$.

Since there are $O(W)$ possible choices for $L$ and $O(W)$ possible choices for $M$, the number of ply covers that we construct is $O(|W|^2)$. Each of these ply covers consists of two independent solutions that can be computed in $O((|W|+|Q|) \log^2 (|W|+|Q|))$ time using the strategy of Case 1.  Hence the overall running time is $O((|W|^3+|W|^2|Q|) \log^2 (|W|+|Q|))$.

%For a point $a$ in $Q$ and let $L_a$ be a vertical line through $a$. Let $Q'_a$ and $Q''_a$ be the points to the left halfplane and right halfplane of $L_a$, respectively. Let $W'_a$ and $W''_a$ be the set of squares of $W$ that lie to the left halfplane and right halfplane of $L_a$, respectively. 
 
%\subsubsection{Case 3 (Opposite   corners of $R$ intersect  all the squares in $W$):}

%\subsubsection{Case 4 (Three   corners of $R$ intersect  all the squares in $W$):}

\subsubsection*{Case 3 (Either two opposite corners or at least three corners  of $R$ intersect the squares in $W$)} 
Let $S$ be a minimum ply cover of $Q$ such that all the squares in $S$ are necessary (i.e., removing a square from $S$ will fail to cover $Q$). 

%If two consecutive corners of $R$ intersects all the squares of $S$, then we could try all of the four possible choices for two consecutive corners and determine the minimum ply cover using Case 2. Therefore, at least three corners of $R$ must intersect some squares in $S$.

%Assume now that at least three corners of $R$ intersect some squares in $S$. We now consider the general scenario. 
Let $c_1,c_2,c_3,c_4$ be the top-left, top-right, bottom-right and bottom-left corners of $R$, respectively. Let $W_i$, where $1\le i \le 4$, be the squares of $W$ that contain $c_i$. Similarly, let $S_i$ be the subset of squares in $S$ that contain $c_i$.

%\subsubsection*{Case 3.1 (A corner of $R$ determines the minimum ply cover number, i.e., it intersects $\ply(S)$ squares of $S$)} 

By Lemma~\ref{lem:disjoint}, one can remove at most four squares from $S$ such that the squares of $S_i$ do not intersect the squares of $S_{(i\bmod 4) +1}$. Consequently, we now have only the following possible  scenarios. 

%\begin{enumerate}
    %\item []
    \smallskip\noindent\textsc{\bf Diagonal:} The squares of $S_i$ do not intersect the squares of $S_{(i\bmod 4) +1}$. The squares of  $S_1$ may intersect the squares of $S_3$, but the squares of $S_2$ do not intersect the squares of $S_4$ (or, vice versa).  Figures~\ref{fig:4-corner}(a) and (b) illustrate such scenarios.%\\
    
    %\item [] 
    \smallskip\noindent\textsc{\bf Disjoint:} If two squares intersect, then they belong to the same set, e.g., Figure~\ref{fig:4-corner}(c).
    \smallskip 
%\end{enumerate} 

We will compute a  minimum ply cover in both scenarios. However,  considering the squares we deleted, the size of the final ply cover we compute may be at most four more than the minimum ply cover.
\begin{figure}
    \centering
    \includegraphics[width=\textwidth]{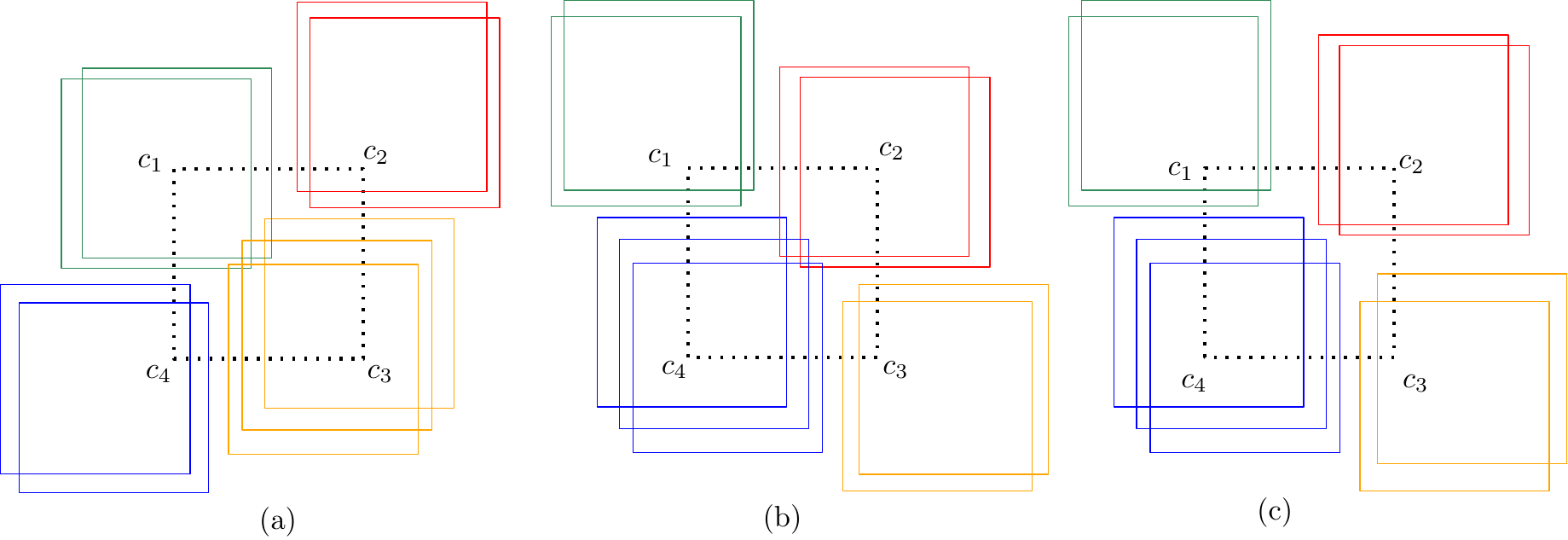}
    \caption{Illustration for the scenarios that may occur after applying Lemma~\ref{lem:disjoint}:     (a)--(b) \textsc{Diagonal}, and (c) \textsc{Disjoint}. }
    \label{fig:4-corner}
\end{figure}

\subsubsection*{Case 3.1 (Scenario \textsc{Diagonal})} We now consider the scenario \textsc{Diagonal}.  Our idea is to  perform a   search on the objective function to determine the minimum ply cover number. Let $k$ be a guess for the minimum ply cover number. Since in this case we know that the ply is determined by a corner of $R$, we ask whether there is a set of squares that cover all the points of $Q$ such that no corner of $R$ intersects more than $k$ squares and the ply of the set is also bounded by $k$. We will use a dynamic program to determine such a set of squares (if exists).

In general, by $T({r}, k_1,k_2,k_3,k_4)$  we denote the problem of finding a minimum ply cover for the points in  a rectangle ${r}$ such  that the ply is at most $\max\{k_1,k_2,k_3,k_4\}$,  and each corner $c_i$ respects its budget $k_i$, i.e., $c_i$ does not intersect more than $k_i$ squares. We will show that ${r}$ can always be expressed as a region bounded by at most four squares in $W$. $T$ returns a feasible ply cover if it exists. To express the original problem, we add four dummy squares in $W$ determined by the four sides of $R$ such that they lie  outside of $R$. Thus ${r} = R$ is the region bounded by the four dummy squares.

Assume without loss of generality that a square $A\in S_4$ intersects a square $B\in S_2$, as shown in Figure~\ref{fig:dp}(a). We assume $A$ and $B$ to be in the minimum ply cover and try out all such pairs. %
We first consider the case when $k\le 3$ and the minimum ply cover already contains $A$ and $B$. We enumerate all $O(|W|^4)$ possible options for $S_2$ and $S_4$ with $\ply(S_2\cup S_4)\le k$ and for each option, we  use Case 1 to determine  whether $\ply(S_1)$ and $\ply(S_3)$ are both upper bounded by $k$. We thus compute the solution to $T({r}, k_1,k_2,k_3,k_4)$  and store them in a table $D({r}, k_1,k_2,k_3,k_4)$, which takes $O ((|W|^5+|W|^4|Q|)\log^2(|W|+|Q|))$ time.

\begin{comment}
{\color{blue} We first examine all potential solutions where $S_2$ and $S_4$ each  contains exactly two squares and all of these four squares mutually intersect, e.g. Figure~\ref{simple}. Since the rectangles of $S_1$ do not intersect the rectangles of $(S_2\cup S_4)$, we can use Case 1 to find the best way to cover (if a cover exists) the points in the top left region, which are not yet covered by the rectangles of $(S_2\cup S_4)$, by the squares that do not intersect $(S_2\cup S_4)$. We precompute the solution and store it in a table $D(A,A',B,B')$, where $A,A'$ are the squares  in $S_4$ and $B,B'$ are the squares in $S_2$. Since there are $O(|W|^4)$ such options, processing all these choices would require $O ((|W|^5+|W|^4|Q|)\log^2(|W|+|Q|))$ time.} 

{\color{blue}We now describe an algorithm to determine the solution for  the scenario \textsc{Diagonal}. We will use the table $D$ as a subroutine.}
\end{comment}

We now show how to decompose $T({r}, k_1,k_2,k_3,k_4)$ into two subproblems assuming that the minimum ply cover already contains $A$ and $B$. {We will use the table $D$ as a subroutine.}
\begin{figure}[h]
    \centering
    \includegraphics[width=\textwidth]{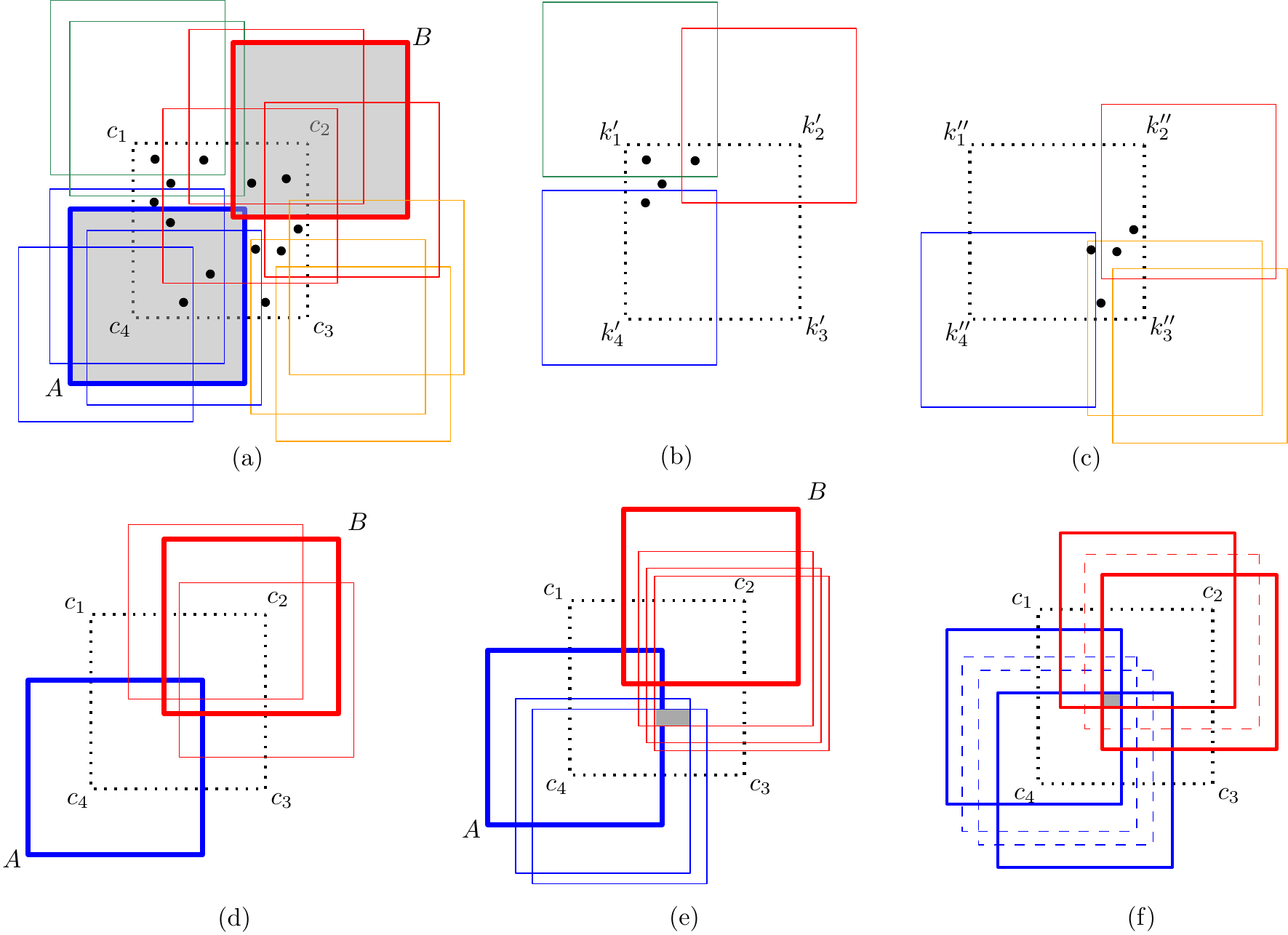}
    \caption{Illustration for the dynamic program. (a)--(c) Decomposition into subproblems. (d)--(f) Illustration for the  $(k+1)$ mutually intersecting squares. The dashed squares can be safely discarded. }
    \label{fig:dp}
\end{figure}

The first subproblem consists of the points that lie above $A$ and to the left of $B$, e.g., Figure~\ref{fig:dp}(b). We refer to these set of points by $Q_1$. The corresponding region ${r}'$ is bounded by four squares: $A$, $B$, and the two (dummy) squares from ${r}$. %We denote these point set by $Q_1$.
We now describe the squares that need to be considered  to cover these points. 

\begin{itemize}
    \item Note that for \textsc{Diagonal}, no square in $S_1$ intersects $A$ or $B$, hence we can only focus on  the squares of $W_1$ that do not intersect $A$ or $B$. 
    \item The squares of $S_2$ that do not intersect $Q_1$ are removed. The squares of $S_2$ that contains the bottom left corner of $B$ are removed because including  them will make  $B$ an  unnecessary square in the cover to be constructed. 
    \item Similarly, the squares of $S_4$ that do not intersect $Q_1$ or contains the top right corner of $A$ are removed. 
    \item No square in $S_3$ needs to be considered since to cover a point of $Q_1$ it must  intersect $A$ or $B$, which is not allowed in \textsc{Diagonal}.
\end{itemize}

The second subproblem  consists of the points that lie below  $B$ and to the right of $A$, e.g., Figure~\ref{fig:dp}(c). The corresponding region ${r}''$ is bounded by four squares: $A$, $B$, and the two squares from ${r}$. We denote these points by $Q_2$. The squares  to be considered can be described symmetrically. 

Let $W'$ and $W''$ be the set of squares considered to cover $Q_1$ and $Q_2$, respectively. By the construction of the two subproblems, we have $Q_1\cap Q_2 = \varnothing$ and $W'\cap W'' = \varnothing$.

For each corner $c_i$, let $k_i'$ and $k_i''$ be the budgets allocated for $c_i$ in the first and the second subproblems, respectively. Since we need to ensure that the ply of the problem $T$ is at most $k=\max\{k_1,k_2,k_3,k_4\}$ and each corner $c_i$ respects its budget $k_i$, we need to carefully distribute the budget among the subproblems when constructing the recurrence formula. Furthermore, let $S'_2$ and $S'_4$ be the sets of  squares corresponding to $c_2$ and $c_4$ that are selected as the solution to the first subproblem. Similarly, define $S''_2$ and $S''_4$ for the second subproblem. We now have the following  recurrence formula.

%\begin{align}
%T({r},k_1,k_2,k_3,k_4) = \min\limits_{\substack{  \{A,B\in W: A\cap B \not = \varnothing\} \\ k'_1+k''_1 = k_1, k'_3+k''_3 = k_3, \\ k'_2+k''_2 = k_2-1, k'_4+k''_4 = k_4-1}} \substack{T({r}',k'_1,k'_2,k'_3,k'_4) +  T({r}'',k''_1,k''_2,k''_3,k''_4) + 2}
%\end{align}

\begin{align*}
T({r},k_1,k_2,k_3,k_4) = \min\limits_{\substack{  \{A\in S_4,B\in S_2: A\cap B \not = \varnothing\} \\ k'_1 = k''_1 = k_1, k'_3 = k''_3 = k_3, \\ k'_2+k''_2 = k_2-1, \\k'_4+k''_4 = k_4-1}}  \begin{cases} 
\substack{{T({r}',k'_1,k'_2,k'_3,k'_4) \cup}\\
T({r}'',k''_1,k''_2,k''_3,k''_4) \cup \{A,B\}} &\text{, if $\delta \le k$  } \\\\
\substack{ {T({r}',k'_1,k'_2,k'_3,k'_4) \cup }\\  T({r}'',k''_1,k''_2,k''_3,k''_4) \cup\beta} &\text{, if $\delta > k$ and $k\ge 4 $} \\\\
\substack{D({r},k_1,k_2,k_3,k_4)} &\text{, if $\delta > k$ and  $k \le 3$}%\\
%\substack{D(S_2,S_4)} &\text{, if $k \le 2$} 
 \end{cases}
\end{align*}
Here $\delta$ is the   ply of $(S'_2\cup S'_4 \cup S''_2 \cup S''_4\cup A \cup B)$ and  $\beta$ is  the set of squares that remain after discarding unnecessary squares  from $(S'_2\cup S'_4 \cup S''_2 \cup S''_4\cup A \cup B)$.  It may initially appear that we must have $k'_3=0$ in $T({r}',k'_1,k'_2,k'_3,k'_4)$ and  $k''_1=0$ in $T({r}'',k''_1,k''_2,k''_3,k''_4)$; however, since $S_1$ and $S_4$ are disjoint, we set $k'_3 = k_3$ and $k''_1 = k_1$ for simplicity.

If $\delta \le k$, then the union of $\{A,B\}$ and the squares obtained from the two subproblems must have a ply of at most $k$ for the following two reasons.  First, the squares of $S_1=S'_1\cup S''_1$ (similarly, $S_3$) cannot intersect the squares of $S_2\cup S_4 =S'_2\cup S''_2\cup S'_4\cup S''_4$. Second, by the budget distribution, the ply of $S_1$ can be at most $k_1\le k$  and the ply  of $S_3$ can be at most $k_3\le k$.

If $\delta >k$ and $k \le 3$, then each of $S_1$, $S_2$, $S_3$, $S_4$ can have at most three  rectangles. We can look it up using the table $D({r},k_1,k_2,k_3,k_4)$.

If $\delta >k \ge 4$, then we can have $k+1$ mutually intersecting squares and in  the following we show how to construct a solution with ply cover at most $k$ respecting the budgets, or to determine whether no such solution exists.  

If  $T({r}',k'_1,k'_2,k'_3,k'_4)$ and  $T({r}'',k''_1,k''_2,k''_3,k''_4)$ each returns a feasible solution, then we know that $(k+1)$ mutually intersecting squares can neither appear in $S'_2\cup S'_4$ nor in $S''_2\cup S''_4$. For example, if $k=4$, then the five thin squares shown in Figure~\ref{fig:dp}(e) cannot be a solution to  $T({r}'',k''_1,k''_2,k''_3,k''_4)$. Therefore, these $k+1$ mutually intersecting squares must include either both $A$ and $B$, or at least one of   $A$ and $B$. We now consider the following options.

%  only with the following configurations. %In each of these configurations 
\begin{itemize}
    
    \item \textbf{Option 1:} $S_4$ and $S_2$ each contains at least two squares that %that  intersect at least $k-1$ other squares of $(S'_2\cup S''_2\cup B)$. 
    belong to the set of $k+1$ mutually intersecting squares.
    Since the region created by the $k+1$ mutually intersecting squares is a rectangle, as illustrated in Figure~\ref{fig:dp}(f), we can keep only the  squares that determine the boundaries of this rectangle to obtain the same  point covering. 
    
    After discarding the unnecessary squares, we only have $\beta$ squares where $|\beta| = 4\le k$. Thus the ply of the union of $S_1\cup S_3$ and the remaining $\beta$ squares is at most $\beta \le k$. Hence we can obtain an affirmative solution by taking $  T({r}',k'_1,k'_2,k'_3,k'_4) \cup  $ $  T({r}'',k''_1,k''_2,k''_3,k''_4) \cup \beta  $.

     \item \textbf{Option 2:} $S_4$ only contains $A$ and $A$  intersects all $k$ squares of $S'_2\cup S''_2\cup B$. Since the $k+1$ mutually intersecting region is a rectangle, as illustrated in Figure~\ref{fig:dp}(d), we can keep only the  squares that determine the boundaries of this rectangle to obtain the same  point covering. After discarding the unnecessary squares, we only have  $\beta$ squares where $|\beta| = 3\le k$. Hence we can handle this case in the same way as in Option 1.
    
    %If $k \ge 3$, then $|S'_2\cup S''_2\cup B|\ge 3$. After discarding the unnecessary rectangles, we only have  $\beta = 3\le k$ rectangles. Thus     the ply of the union of $(S_1\cup S_3)$ and the remaining $\beta$ rectangles is at most $\beta \le k$. We thus obtain the required solution by taking $\max\{ T({r}',k'_1,k'_2,k'_3,k'_4) ,  T({r}'',k''_1,k''_2, k''_3,k''_4) , \beta \}$.      If $k \le 2$, then each of $S_1$, $S_2$, $S_3$, $S_4$ can have at most two rectangles. We can look it up using the table $D({r},k_1,k_2,k_3,k_4)$.
    
    \item \textbf{Option 3:} $S_2$ only contains $B$ and $B$  intersects all $k$ squares of $S'_4\cup S''_4\cup A$. This case is symmetric to Option 2. 
\end{itemize}

\smallskip
\noindent
In the base case, we either covered all the points, %have $k=O(1)$,
or we obtain a set of problems of type Case 1 or of Scenario \textsc{Disjoint} (Case 3.1.2).
 The potential base cases corresponding to Case 1 are formed by guessing $O(|W|^2)$ pairs of intersecting squares from opposite corners, as illustrated in Figure~\ref{base}(a).
  The potential $O(|W|^4)$ base cases corresponding to Scenario \textsc{Disjoint} are formed by two pairs of intersecting squares from opposite corners, as illustrated in Figure~\ref{base}(b).

\begin{figure}[h]
    \centering
    \includegraphics[width=.8\textwidth]{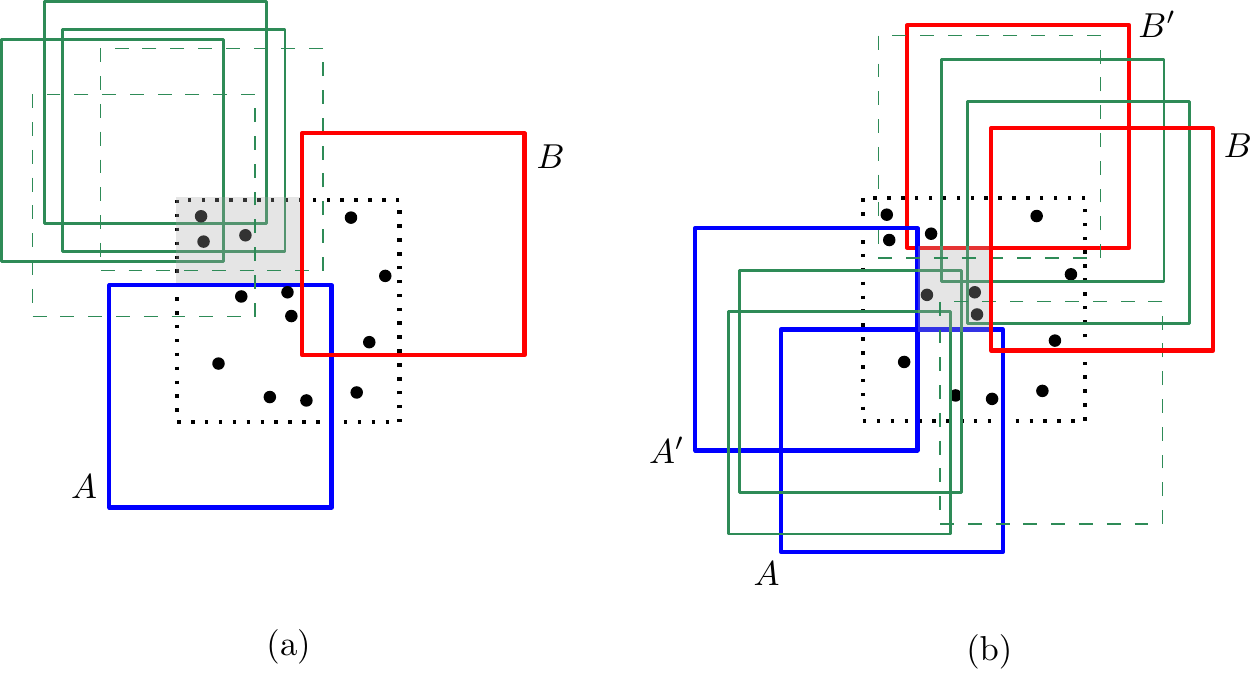}
    \caption{Illustration for the base cases, where the region corresponding to the base cases are shown in gray. (a) The base case corresponds to Case 1, where we ignore the squares that intersect the chosen squares $A$ and $B$. (b) An example of the base case that corresponds to scenario \textsc{Disjoint}, where we need to construct a solution such that no two squares from opposite corners intersect. We ignore all the squares that intersect the chosen squares $A$ and $B$, or $A'$ and $B'$, as well as those that makes any of them unnecessary. }
    \label{base}
\end{figure}

%We precompute all such problems %by guessing $O(|W|^2)$ pairs of squares from opposite corners such that the solution to the base cases can be retrieved in $O(1)$ time. This 
The precomputation of the base cases takes  $O(|W|^4f(|W|,|Q|))$ time, where  $f(|W|,|Q|)$ is the time to solve a problem of type Case 1 and of Scenario \textsc{Disjoint}. We will discuss the details of $f(|W|,|Q|)$  in the proof of Theorem~\ref{thm:unitSquare}.

Since $r$ is determined by at most four squares (e.g., Figure~\ref{base}), and since there are four budgets, the solution to the subproblems can be stored in a dynamic programming table of size $O(|W|^4k^4)$. Computing each entry requires examining $O(|W|^2)$ pairs of squares. Thus the overall running time becomes $O(|W|^6k^4+|W|^4f(|W|,|Q|))$.

\subsubsection*{Case 3.2 (Scenario \textsc{Disjoint})}  In this case, we can find a sequence of empty rectangles $\sigma=(e_1,e_2,\ldots)$ from top to bottom such that they do not intersect any square of $S$, as illustrated in Figure~\ref{fig:dp2}(a)--(b). 

The top side of the first rectangle $e_1$  is determined by the dummy square at the top side of $R$. The left side of $e_1$ is determined by the topmost  square $A_1$ of $S_1$, or if $S_1$ is empty, then by the dummy square at the left side of $R$. Similarly, the right side is determined by the topmost square $B_1$ of $S_2$ or a dummy square.   The bottom side of $e_1$ is determined by the    lowest horizontal line that touches both $A_1$ and  $B_1$. 

Each subsequent empty rectangle $e_i$ in $\sigma$ is determined as follows. The top side of $e_i$ is determined by the bottom side of $e_{i-1}$. The left side of $e_i$ is determined by a square from $A_i\in S_1\cup S_4$ or by the dummy square  at the left side of $R$. The right  side of $e_i$ is determined by a square from $B_i\in S_2\cup S_3$ or by the dummy square at the right side of $R$. The bottom side of $e_i$ is determined by the    lowest horizontal line that touches both  $A_i$ and $B_i$.

\begin{figure}[h]
    \centering
    \includegraphics[width=\textwidth]{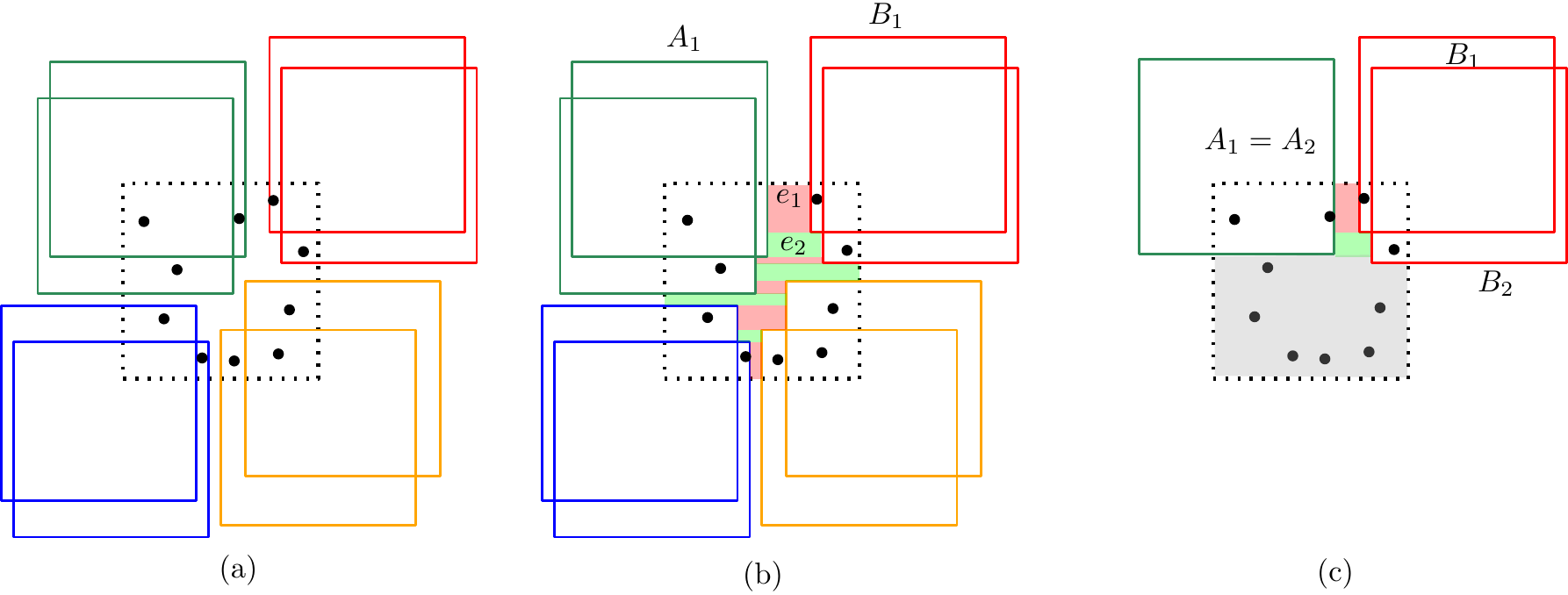}
    \caption{Illustration for the dynamic program.   }
    \label{fig:dp2}
\end{figure}

Similar to Case 3.1, we show how to solve the problem  $T({r},k_1,k_2,k_3,k_4)$, but here the region ${r}$ can be expressed with at most two squares, and is typically L-shaped. The $k_i$ corresponds to the budget given for the corner $c_i$ of $R$, where $1\le i\le 4$. Let $A,B$ be two squares,  where $A\in S_1\cup S_4$ and $B\in S_2\cup S_3$. Let $\ell$ be the lowest horizontal line that touches both $A$ and $B$, then ${r}$ is the region of $R$ that lies below $\ell$ and excludes $A$ and $B$, as shown by the shaded region in Figure~\ref{fig:dp2}(c).

We check all possibilities for $A_1\in S_1$ and $B_1\in S_2$. Then to decompose $T({r},k_1,k_2,k_3,k_4)$, we guess the next empty rectangle by choosing squares  $A_i\in S_1\cup S_4$ and $B_i\in S_2\cup S_3$. %Note that here $A_i$ may coincide with $A_{i-1}$ and $B_i$ may coincide with $B_{i-1}$. % but not both. 
 We must guess at least one new square, i.e., we may have $A_i = A_{i-1}$ and $B_i = B_{i-1}$ but not both simultaneously.
If $A_i\in S_1$ and $A_{i-1}\not = A_i$, then the right side of $A_i$ must intersect the bottom side of $A_{i-1}$ (otherwise, $A_i$ would be unnecessary).   Similarly, if  $B_i\in S_2$ and    $B_{i-1}\not = B_i$, then the left side of $B_i$ must intersect the bottom side of $B_{i-1}$. If  $A_{i-1}\in S_1$ and $A_i\in S_4$, then $A_{i-1}$ and $A_i$ must be disjoint. Similarly, if $B_{i-1}\in S_2$ and $B_i\in S_3$, then $B_{i-1}$ and $B_i$ must be disjoint. We refer to such a pair of squares $A_i$ and $B_i$ to be a \emph{compatible pair} and define the recurrence relation as follows. 
 
\begin{align}
T({r},k_1,k_2,k_3,k_4) = \min\limits_{\substack{  \{
A_i,B_i \text{ are} \\ 
\text{ compatible}\}
}}  \begin{cases}
   T({r}',k'_1,k'_2,k'_3,k'_4),       & \text{if both } A_i, B_i \text{ are dummy}  \\
     & \text{or, } B_i \text{ is dummy and }A_i=A_{i-1},  \\
          & \text{or, } A_i \text{ is dummy and }B_i=B_{i-1},  \\ 
   T({r}',k'_1,k'_2,k'_3,k'_4) \cup\\  (\{A_i, B_i\} \setminus \{A_{i-1}, B_{i-1}\})   & \text{, otherwise}.  \\
  \end{cases} 
  \nonumber
\end{align}
Let $\ell$ be the lowest horizontal line that touches both $A_i$ and $B_i$, then ${r'}$ is the region of $R$ that lies below $\ell$ and excludes $A_i$ and $B_i$. The budgets $k'_1$ and $k'_4$ are set depending on whether $A_i$ belongs to $S_1$ or $S_4$, as well as whether $A_i$ coincides with $A_{i-1}$. The budgets $k'_2$ and $k'_3$ are set similarly. In the base case, either we find that the points are covered without exceeding the budget, or run out of budget with points remain uncovered, or we cannot find a compatible pair to form the next empty rectangle. 

If the dynamic program returns an affirmative solution, then it corresponds to one where each corner $c_i$ of $R$ hits at most $k_i$ squares. Since the sets $S_1,S_2,S_3,S_4$ are disjoint, it is straightforward to verify that the minimum ply cover number of $T(r,k_1,k_2,k_3,k_4)$ is at most $k = \max\{k_1,k_2,k_3,k_4\}$. 

Since $r$ is determined by two squares, and there are four budgets, the solution to the subproblems can be stored in a dynamic programming table of size $O(|W|^2k^4)$. Computing each entry requires examining $O(|W|^2)$ pairs of squares, and this takes $O(|W|^2)$  time assuming that we preprocessed  whether two squares are compatible. It is straightforward to do the preprocessing for various choices of $A_{i-1},A_i,B_{i-1}B_i$ in $O(|W|^4\log |Q| + |Q|\log |Q|)$ time, where the term $|Q|\log |Q|$ is for constructing the range tree to process the empty rectangle queries. Thus the overall running time becomes $O(|W|^4k^4 + |W|^4\log |Q| + |Q|\log |Q|)$.

\jyoti{
We now have the following theorem that summarizes the result of this section. % (see Appendix~\ref{app:unitSquare} for the run-time analysis). 
}
\begin{theorem}\label{thm:unitSquare}
Given a set $Q$ of points inside a unit square $R$ and a set $W$ of axis-parallel unit squares, a ply cover of size $4+\ply^*(Q,W)$ %for $Q$ using $W$
can be computed in $O(( |W|^8(k^*)^4 + |W|^8\log |Q| + |W|^4|Q|\log |Q| )$ $\log k^*)$ %$O(|P|+|W|^{10} \log k)$ %polynomial
 time, where $k^* = \ply^*(Q,W) \leq \min\{|Q|,|W|\}$.
\end{theorem}
\begin{proof}
We first compute the ply cover for each $W_i$ (Case 1), where $W_i$ is the set of squares containing the corner $c_i$ of $R$, where $1\le i\le 4$. This takes $O((|W|+|Q|)\log^2(|W|+|Q|))$ time.

We then compute the ply cover for each pair of consecutive corners (Case 2). This takes $O((|W|^3+|W|^2|Q|)\log^2(|W|+|Q|))$ time.

We next compute a ply cover of size at most $4+\ply^*(Q,W)$ considering the general case  (Case 3). % when a corner of $R$ determines the minimum ply (Case 3.1). 
We guess at most 4 squares such that deleting them will guarantee some separation between  the squares in the minimum ply cover solution. There are two scenarios \textsc{Diagonal} and \textsc{Disjoint}. Solving \textsc{Disjoint} takes $O(|W|^4k^4 + |W|^4\log |Q| + |Q|\log |Q|)$ time. Solving \textsc{Diagonal} takes  $O(|W|^6k^4 + |W|^4f(|W|,|Q|) )$ time, where  $f(|W|,|Q|)$ is the time to solve a problem of type Case 1 and of Scenario \textsc{Disjoint}. Therefore, the running time for \textsc{Diagonal} is 
\begin{align*}
     & O(|W|^6k^4 + |W|^4 \cdot (  |W|^4k^4 + |W|^4\log |Q| + |Q|\log |Q|  + (|W|+|Q|)\log^2(|W|+|Q|)   )  )\\
     &\subseteq O(|W|^6k^4 +    |W|^8k^4 + |W|^8\log |Q| + |W|^4|Q|\log |Q|       )\\
         &= O(  |W|^8k^4 + |W|^8\log |Q| + |W|^4|Q|\log |Q|       ).
         %&\in O(  |W|^6k^4 + |W|^6|Q| \log |Q|        ). %if Q=W^5??
\end{align*}

%Finally, we compute the ply cover considering the case that the region that determines the ply lies strictly interior to $R$ (Case 3.2). In this case the running time is $O(|W|^5+|W|^4|Q|)\log^2 (|W|+|Q|))$.  

Let $r^*$ be the minimum ply cover number of the instance obtained after removing at most 4 squares. The algorithm does not initially know $r'^*$. % = \ply^*(Q,W)$.
To minimize running time, start with an initial value $r \in O(1)$. If the algorithm fails to find a solution because $r$ is too small, double the value of $r$ and repeat. A sufficiently large value $r$ will be discovered within $O(\log r^*)$ iterations, leaving an interval of possible values for $r^* \in \{r/2, r/2+1, \ldots, r\}$. The value $r^*$ can be found by applying binary search on the interval and running our algorithm on the corresponding value of $r$ at each iteration. These two steps increase our running time by a factor of $O(\log r^*)$. Since a minimum ply cover can include at most all squares in $W$, and, in the worst case, each point in $Q$ is covered by a unique square in $W$, we get that $r^* \leq \min\{|Q|, |W|\}$.

We take the minimum of all ply cover to compute the final ply cover for $P$. If the squares in the minimum ply cover is intersected by one corner of $R$ (or two consecutive corners of $R$), then the ply cover we compute in Case 1 (or Case 2) will be able to identify a minimum ply cover. Otherwise, the size of the ply cover we compute in Case 3 will be at most four more than the minimum ply cover number. %Hence the running time is $O(|W|^8k^4 + |W|^8\log |Q| + |W|^4|Q|\log |Q|   )$. 
\end{proof}

\subsection{Covering a General Point Set}
\label{gen}
%Let $P$ be a set of $n$ points interior to a unit square $S$. Given a set $U$ of unit squares intersecting $S$, the \emph{minimum ply cover of $P$} is a subset of $U$ that covers all the points in $P$ and has the minimum ply over all subsets of $U$.

%The optimal ply is either at a corner or the max ply at a corner and some constant.

%In the first case we apply the dynamic programming directly.

%In the second case we first guess the rectangles to be deleted and then apply the dynamic programming.

%For the dynamic programming we guess two from opposite corners and then balance the allocation so that the overall bound remains the same.

Given a set $P$ of points and a set $U$ of axis-parallel unit squares, both in $\mathbb{R}^2$, we now give a polynomial-time algorithm that returns a ply cover of $P$ with $U$ whose cardinality is at most \jyoti{$(8+\varepsilon)$} times the minimum ply cover number of $P$ with $U$. Recall that our algorithm partitions $P$ along a unit grid and applies Theorem~\ref{thm:unitSquare} iteratively at each grid cell to select a subset of $U$ that is a minimum ply cover for the grid cell. Elements of $U$ selected to cover points of $P$ in a given grid cell overlap neighbouring grid cells, which can cause the ply to increase in those neighbouring cells; Lemma~\ref{lem:approx8} allows us to prove Theorem~\ref{thm:approxAlg} and Corollary~\ref{cor:approxAlg}, showing that the resulting ply is at most \jyoti{$(8+\varepsilon)$} times the optimal value.

Partition $P$ using a unit grid. I.e., each cell in the partition contains $P \cap [i, i+i) \times [j, j+1)$, for some $i,j \in \mathbb{Z}$. Each grid cell has eight grid cells adjacent to it. Let $C_1, \ldots, C_4$ denote the four grid cells that are its diagonal neighbours. See Figure~\ref{fig:grid1}.

\begin{figure}[h]
    \centering
    \includegraphics[width=0.2\textwidth]{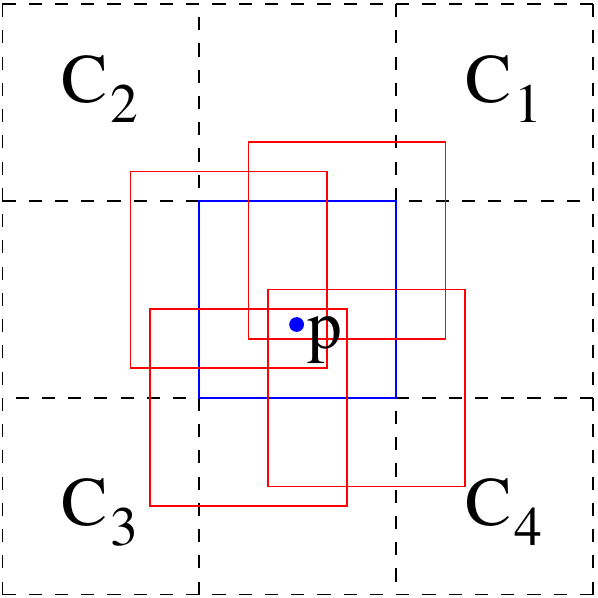}
    \caption{The grid cells $C_1, \ldots, C_4$ are the diagonal grid neighbours of the blue grid cell. If any point $p$ in the blue grid cell is covered by squares (red) that intersect the grid cells $C_1$--$C_4$, then these same squares cover the entire blue grid cell (Lemma~\ref{lem:approx8}).}
    \label{fig:grid1}
\end{figure}
 
\begin{lemma}\label{lem:approx8}
If any point $p$ in a grid cell $C$ is contained in four squares, $\{S_1, \ldots, S_4\} \subseteq U$, such that for each $i \in \{1, \ldots, 4\}$, $S_i$ intersects the cell $C_i$ that is $C$'s diagonal grid neighbour, then $C \subseteq S_1 \cup S_2 \cup S_3 \cup S_4$. 
\end{lemma}

\begin{proof}
Choose any grid cell $C$ and any point $p \in C$. Let $C_1, \ldots, C_4$ denote the four grid cells that are $C$'s diagonal grid neighbours. Suppose $\{S_1, \ldots, S_4\} \subseteq U$ is a set of axis-aligned squares such that for each $i\in \{1,\ldots, 4\}$, $C_i \cap S_i \neq \varnothing \wedge p \in S_i$; i.e., each diagonal grid cell intersects some square that covers $p$. Observe that $S_1$ covers the region of $C$ above and to the right of $p$. Similarly, $S_2$ covers the region of $C$ above and to the left of $p$. By analogous observations for $S_3$ and $S_4$ with respect to $C_3$ and $C_4$, it follows that $S_1 \cup S_2 \cup S_3 \cup S_4$ covers $C$.
\end{proof}

%\needtobechanged{Jyoti: I checked till this point.}

\begin{theorem}\label{thm:approxAlg}
Given a set $P$ of points and a set $U$ of axis-parallel unit squares, both in $\mathbb{R}^2$, a ply cover of $P$ using $U$ can be computed in $O(( |U|^8 (k^*)^4 + |U|^8\log |P| + |U|^4|P|\log|P|) \log k^*)$ %$O(|P|+|U|^{10}\log |U|)$
time whose ply is at most $8k^*+32$, where $k^* = \ply^*(P,U) \leq \min\{|P|,|U|\}$ denotes the minimum ply cover number of $P$ by $U$. 
\end{theorem}

\begin{proof}
Partition $P$ along a unit grid and apply Theorem~\ref{thm:unitSquare} iteratively to find a near minimum ply cover for each grid cell. 
Let $S \subseteq U$ denote the union of the sets of squares selected to cover the points of $P$ in each grid cell.
Mark each grid cell that contains a point of $P$.
Select any marked grid cell $C$ whose eight grid neighbours are marked. 
Observe that an axis-parallel square $S_i$ that intersects $C$ can intersect one or more of the eight grid cells that are $C$'s neighbours in the grid, but $S_i$ cannot intersect any other grid cell outside these nine grid cells. 

%case 1 not necessarily points in P?
{\bf Case 1.} Suppose no point in $C$ is covered by squares that cover points in each of the four grid cells that are diagonal grid neighbours to $C$. For every point $p$ in $C$, there is some diagonal grid neighbour $C'$ of $C$ such that $p$ is not covered by any square that covers points in $C'$. Consequently, only seven grid neighbours of $C$ can contribute to the ply of $p$. 

{\bf Case 2.} Suppose some point $p$ in $C$ is covered by squares that cover points in each of the four grid cells that are diagonal grid neighbours to $C$. By Lemma~\ref{lem:approx8}, these squares cover $C$. Consequently, if such a point $p$ in $C$ exists, then the squares of $U$ selected to cover $C$ that were not selected to cover any of its neighbouring grid cells are not required; remove these from the solution set $S$. Remove the mark from the grid cell $C$. Select any remaining marked grid cell whose eight grid neighbours are marked, and repeat until no such grid cell remains. Upon termination, the ply in any grid cell $C$ can be at most eight times the ply in its neighbouring cells, which is at most $8 \cdot(k^* + 4) = 8k^* + 32$. 

Since each square in $U$ intersects at most four grid cells, the total time for finding a minimum ply cover for grid cells sums to $O(( |U|^8 (k^*)^4 + |U|^8\log |P| + |U|^4|P|\log|P|) \log k^*)$, where $k^* = \ply^*(P,U)$. The next phase, iteratively marking cells and removing squares from the solution set $S$, takes time $O(|P| + |U|)$.
\end{proof}

 A set $U$ of axis-parallel squares in $\mathbb{R}^2$ determines $O(|U|^2)$ connected regions or faces.  If such a face is non-empty (i.e., contains some points from $P$), then we can pick an arbitrary representative point of $P$ in that face and discard the rest. It is straightforward to perform this computation in $O(|U|^2|P|)$ time. Observe that the ply cover of the resulting instance would be the  same  as the ply cover of the input instance. Hence we can simplify the time complexity of Theorem~\ref{thm:approxAlg} by assuming $|P| \in O(|U|^2)$.

\begin{corollary}\label{cor5}
Given a set $P$ of points and a set $U$ of axis-parallel unit squares, both in $\mathbb{R}^2$, a ply cover of $P$ using $U$ can be computed in 
$O(|P|+|U|^8 (k^*)^4\log k^* + |U|^8\log |U| \log k^*)$
%$O(|P|+|U|^{10}\log |U|)$
time whose ply is at most 
$8k^*+32$, where $k^* = \ply^*(P,U) \leq \min\{|P|,|U|\}$ denotes the minimum ply cover number of $P$ by $U$. 
\end{corollary}

Biedl et al.~\cite{DBLP:journals/comgeo/BiedlBL21} gave a 2-approximation that runs in polynomial time when $k^* \in O(1)$. We can use their algorithm when $k^*$ is small, and use our algorithm for larger values of $k^*$.%, which reduces the approximation factor to nine.

\begin{corollary}\label{cor:approxAlg}
Given a set $P$ of points and a set $U$ of axis-parallel unit squares, both in $\mathbb{R}^2$, a ply cover of $P$ using $U$ can be computed in 
polynomial
%$O(( |U|^8 (k^*)^4 + |U|^8\log |P| + |U|^4|P|\log|P|) \log k^*)$ %$O(|P|+|U|^{10}\log |U|)$
time whose ply is at most $(8+\varepsilon)$ times the minimum ply cover number $k^* = \ply^*(P,U)$, for every fixed  $\varepsilon > 0$. % $k^* = \ply^*(P,U).$% \leq \min\{|P|,|U|\}$.
\end{corollary}

\begin{proof}
By Corollary~\ref{cor5}, we can find a ply cover of $P$ by $U$ 
in $O(|P|+|U|^8 (k^*)^4\log k^* + |U|^8\log |U| \log k^*)$
%in $O(( |U|^8 (k^*)^4 + |U|^8\log |P| + |U|^4|P|\log|P|) \log k^*)$ %$O(|P|+|U|^{10}\log |U|)$
time whose ply is at most $8k^*+32$.

{\bf Case 1.} \jyoti{Suppose $\varepsilon k^* \ge 32$. Then %Since ply is integral, 
%$k^*  \geq 32 \Rightarrow 8
$8k^* + 32  \leq (8+\varepsilon) k^*$. 
%\begin{align*}
%&& k^* & \geq 32 \\
%\Rightarrow && 8 k^* + 32 & \leq 9 k^*.
%\end{align*}
%Suppose $k^* > 5$. Since ply is integral, 
%\begin{align*}
%&& k^* & \geq 6 \\
%&& & > 4 \sqrt{2} \\
%\Rightarrow && (k^*)^2 & > 32 \\
%\Rightarrow && k^* & > \frac{32}{k^*} \\
%\Rightarrow && 9 k^* & > 8k^* + \frac{32}{k^*} .
%\end{align*}
%In Case~1, the ply of the solution returned is at most nine times the optimal ply.
}

{\bf Case 2.} \jyoti{Suppose $\epsilon k^* <  32$. 
We apply the 2-approximation algorithm of Biedl et al.~\cite{DBLP:journals/comgeo/BiedlBL21} in $O(|P| \cdot |U|)^{3k^*+1})$ time, which is polynomial since $k^* \in O(1)$.}
%In Case~2, the ply of the solution returned is at most two times the optimal ply.
\end{proof}
 
\section{Conclusion}

In this paper we settled the open question \cite{DBLP:journals/comgeo/BiedlBL21,DBLP:conf/cccg/BiniazL20} of whether there exists a polynomial-time $O(1)$-approximation algorithm for the minimum ply cover problem with axis-parallel unit squares. We   gave a \jyoti{$(8+\varepsilon)$}-approximation polynomial-time algorithm for the problem.     Through careful case analysis, it may be possible to further improve the running time of our approximation algorithm presented in Theorem~\ref{thm:approxAlg}. A natural direction for future research would be to reduce the approximation factor or to apply a different algorithmic technique with lower running time. It would also be interesting to examine whether our strategy can be generalized to find polynomial-time approximation algorithms for other covering shapes, such as unit disks or convex shapes of fixed size. 

%\section{Minimum Ply Covering with Unit Disks}

%In this section we give an  $O(??)$-time algorithm to approximate the \mpcn for unit disks. The idea of the algorithm is the same as in Section~\ref{unitsquare}, but the uniform grid that we use is of finer resolution. In particular, we define a grid such that the cells are of size $\frac{1}{\sqrt{2}}\times \frac{1}{\sqrt{2}}$. 

%We first show how to find a minimum ply cover when the points are bounded inside a $\frac{1}{\sqrt{2}}\times \frac{1}{\sqrt{2}}$  square (Section~\ref{dis}) and then show how an approximate ply cover number can be computed for $P$ (Section~\ref{gendis}).

%\subsection{Points are in a Unit Square}
%\label{dis}

%\subsection{Covering a General Point Set}
%\label{gendis}

%\bibliographystyle{splncs04}

\bibliography{ply}

\end{document}